\documentclass[11pt,a4paper]{article}%
\usepackage{amssymb,amsmath,amsfonts,amsthm,array,bm,color}%
\usepackage{amsmath}%
\usepackage{amsfonts}%
\usepackage{amssymb}%
\usepackage{graphicx}
\providecommand{\U}[1]{\protect\rule{.1in}{.1in}}
\newtheorem{theorem}{Theorem}[section]
\newtheorem{lemma}[theorem]{Lemma}
\newtheorem{corollary}[theorem]{Corollary}
\newtheorem{proposition}[theorem]{Proposition}
\newtheorem{remark}[theorem]{Remark}

\newtheorem{definition}[theorem]{Definition}
\def\<{\langle}
\def\>{\rangle}
\def\d{{\rm d}}
\def\L{\mathcal{L}}

\def\E{\mathbb{E}}
\def\H{{:\!\mathcal H\!:}}
\def\N{\mathbb{N}}
\def\P{\mathbb{P}}
\def\R{\mathbb{R}}
\def\T{\mathbb{T}}
\def\Z{\mathbb{Z}}

\def\eps{\varepsilon}

\begin{document}

%\maketitle
\makeatletter
%'@' is now a normal "letter" for TeX
\renewcommand\theequation{\thesection.\arabic{equation}}
\@addtoreset{equation}{section} \makeatother
%'@' is restored as a "non-letter" character for TeX

\title{Energy conditional measures and 2D turbulence}

\author{Franco Flandoli\footnote{Email: franco.flandoli@sns.it. Scuola Normale Superiore of Pisa, Italy.} \ and
Dejun Luo\footnote{Email: luodj@amss.ac.cn. RCSDS, Academy of Mathematics and Systems Science, Chinese Academy of Sciences, Beijing 100190, China, and School of Mathematical Sciences, University of the Chinese Academy of Sciences, Beijing 100049, China. }}

\maketitle

\begin{abstract}
We show that the invariant measure of point vortices, when conditioning the Hamiltonian to a finite interval, converges weakly to the enstrophy measure by conditioning the renormalized energy to the same interval. We also prove the existence of solutions to 2D Euler equations having the energy conditional measure as invariant measure. Some heuristic discussions and numerical simulations are presented in the last section.
\end{abstract}


%\textbf{MSC2010:} 76C05

\section{Introduction}

For the dynamics of point vortices associated to the 2D Euler equations on the torus $\T^2 = \R^2/\Z^2$, the Lebesgue  measure is invariant. If we also randomize the intensities, a measure like
  \begin{equation}\label{invariant-meas}
  \lambda_N(\d x_1, \ldots, \d x_N, \d\xi_1,\ldots, \d\xi_N)= \d x_1 \ldots \d x_N \mathcal N(\d \xi_1) \ldots \mathcal N(\d \xi_N)
  \end{equation}
is invariant too, where $\mathcal N$ denotes the standard Gaussian distribution on $\R$. Define the Hamiltonian
  \begin{equation}\label{Hamiltonian}
  \mathcal H_N((\xi_1, x_1),\ldots, (\xi_N, x_N))= -\frac1{2N} \sum_{1\leq i\neq j\leq N} \xi_i \xi_j G(x_i-x_j),
  \end{equation}
where $G$ is the Green function on $\T^2$. The coefficient $-\frac12$ here is chosen so that $\mathcal H_N$ converge weakly to the renormalized energy of the white noise, see Proposition \ref{prop-weak} for details. We can rewrite $\mathcal H_N$ as a functional of the point vortices
  $$\omega_N= \frac1{\sqrt N} \sum_{i=1}^N \xi_i \delta_{x_i}.$$
Indeed, if we regard $G$ as a function on $\T^2\times \T^2$ by setting $G(x,y)= G(x-y)$ and $G(x,x)=0$ for all $x,y\in \T^2$, then
  \begin{equation}\label{Hamiltonian-1}
  \mathcal H_N= -\frac12 \<\omega_N\otimes \omega_N, G \>.
  \end{equation}
Let $\mu_N$ be the distribution of $\omega_N$ on $H^{-1-}(\T^2)$ under the measure $\lambda_N$, where $H^{-1-}(\T^2)$ is the intersection of all the Sobolev spaces $H^{s}(\T^2)$ of order less than $-1$. It is proved in \cite[Proposition 21]{F1} that $\mu_N$ is weakly convergent to the enstrophy measure $\mu$, which is supported by $H^{-1-}(\T^2)$.

Recall that $\mu$ is the distribution of the white noise on $\T^2$. Let $\H= \ :\!\!\mathcal H(\omega)\!\!:$ be the renormalized energy of a white noise $\omega$ (see Section 2.2 for its definition); we can regard $\H$ as a random variable defined on the probability space $\big(H^{-1-}(\T^2), \mathcal B\big(H^{-1-}(\T^2) \big), \mu\big)$, where $\mathcal B\big(H^{-1-}(\T^2) \big)$ is the collection of Borel measurable sets. Fix $a,b\in \R$ with $a<b$; from Proposition \ref{prop-support} below, we always have
  $$\mu(\{\H \in [a,b]\}) >0. $$
Therefore, the conditional measure
  \begin{equation}\label{energy-cond-meas}
  \mu^{a,b}(A) = \frac{\mu(A \cap \{\H \in [a,b]\})}{\mu(\{\H \in [a,b]\})}, \quad A\in \mathcal B\big(H^{-1-}(\T^2) \big)
  \end{equation}
is well defined. On the other hand, we shall prove in Proposition \ref{prop-weak} that $\lim_{N\to\infty} \mu_N(\{\mathcal H_N \in [a,b]\}) = \mu(\{\H \in [a,b]\}) >0$, hence we can define in the same way the energy conditional measures for the point vortices:
  \begin{equation}\label{microcanonical-meas}
  \mu_N^{a,b}(A) = \frac{\mu_N(A \cap \{\mathcal H_N \in [a,b]\})}{\mu_N(\{\mathcal H_N \in [a,b]\})}, \quad A\in \mathcal B\big(H^{-1-}(\T^2) \big).
  \end{equation}
The main result in the present paper is

\begin{theorem}\label{main-thm}
The family $\big\{\mu_N^{a,b} \big\}_{N\geq 1}$ of energy conditional measures converge weakly to $\mu^{a,b}$.
\end{theorem}

This result shows the convergence of a class of microcanonical measures. We mention that, when the intensities $\{\xi_i \}_{i\geq 1}$ are i.i.d. centered Bernoulli random variables, Benfatto et al. \cite{BPP} proved that the canonical Gibbs measures of the point vortices, with appropriately regularized Green functions, converge to the Gaussian measure $\mu_{\beta, \gamma}(\d\omega) = {\rm e}^{-\beta H -\gamma E}\,\d\omega$ ($\beta, \gamma >0$, $H$ and $ E$ are the energy and enstrophy functionals), which are invariant for the 2D Euler flow. In the recent work \cite{GrottoRomito}, analogous result was proved without smoothing the Green function; see \cite{GR} for related result concerning the generalised inviscid surface quasi-geostropic equations.

It is worth pointing out, since the interaction in the Euler dynamics is of long range, that the general principle of equivalence of ensembles is not necessarily valid, see \cite{TER} for more discussions on nonequivalence of ensembles. In the mean field regime of Onsager theory it holds true, see \cite{CLMP-1, Eyink}, but in the regime studied here the infinite particle limit is not the Gibbs measure associated to the renormalized energy. It seems that the relevant statistical ensemble is the microcanonical measure, see \cite[p.237]{BV} for detailed explanations.

The results presented here are meant to be fragments of a more general investigation on invariant measures of 2D Euler equations, in the attempt to capture some features of inverse cascade turbulence. Onsager theory, extremely relevant for the explanation of large scale coherent vortex structures, does not provide a description of inverse stationary turbulence; but unfortunately also the regime studied here is not the correct description. In a sense, Onsager theory and the regime considered here are two extremes, both with relevant features, but turbulence is somewhat in between. In the last section we discuss this issue.

The above theorem will be proved in Section 3. To this end, we first make some necessary preparations in Section 2, including the definitions of $\<\omega \otimes \omega, G\>$ and of the renormalized energy $\H= \ :\!\!\mathcal H(\omega)\!\!:$ for a white noise $\omega$; the relation between them will be clarified in Section 2.3. We study in Section 4 the limit behaviors of the correlation functions of the energy conditional measures on the ``flat space'' $(\R\times \T^2)^N$, following some arguments in \cite[Section 5.4]{Lions} (see also \cite{CLMP-1, CLMP-2, Neri} for related results). Based on the results in \cite{F1}, we prove in Section 5 the existence of solutions to the 2D Euler equations having the energy conditional measure $\mu^{a,b}$ as invariant measure. Finally, we present in the last section some heuristical discussions together with numerical simulations of the spectrum functions for point vortices, illustrating the relevance of our results with 2D turbulence.

\section{Preliminary results on the renormalized energy}

In this section we make some preparations regarding the renormalized energy of a white noise $\omega$ on $\T^2$. Firstly, we follow the idea in \cite[Section 2.4]{F1} to define the quantity $\<\omega\otimes \omega, G\>$, where $G$ is the Green function on $\T^2$. Secondly, we recall the definition of the renormalized energy $\H$ using the Galerkin approximation; based on the series expansion of $\H$, we are able to show that its distribution has full support on the real line. Finally, we study the relation between $\<\omega\otimes \omega, G\>$ and $\H$, see Theorem \ref{thm-coincidence}.

\subsection{Definition of $\<\omega\otimes \omega, G\>$ for a white noise $\omega$} \label{sec-2.1}

In this part we follow the approach in \cite[Section 2.4]{F1} (see also \cite[Section 2.2]{FS}) to define the quantity $\<\omega\otimes \omega, G\>$ when $\omega$ is a white noise on $\T^2$. The results below are proved in \cite[Corollary 6]{F1} and we omit the proofs here.

\begin{lemma}\label{lem-1}
\begin{itemize}
\item[\rm(i)] If $\omega: \Theta \to C^\infty(\T^2)'$ is a white noise and $f\in H^{2+}(\T^2\times \T^2)$, then for every $p\geq 1$ there is constant $C_p>0$ such that
  $$\E[|\<\omega\otimes \omega, f\>^p|] \leq C_p \|f\|_\infty^p.$$

\item[\rm(ii)] We have $\E \<\omega\otimes \omega, f\> = \int_{\T^2} f(x,x)\,\d x$.

\item[\rm(iii)] If $f$ is symmetric, then
  $$\E\big[ |\<\omega\otimes \omega, f\> - \E\<\omega\otimes \omega, f\>|^2\big]= 2 \int_{\T^2\times \T^2} f(x,y)^2\,\d x\d y.$$
\end{itemize}
\end{lemma}

Based on these facts we can give a definition of $\<\omega\otimes \omega, G\>$ when $\omega$ is a white noise on $\T^2$.

\begin{proposition}\label{prop-1}
Let $\omega: \Theta \to C^\infty(\T^2)'$ be a white noise. Assume that $G_n\in H^{2+}(\T^2\times \T^2)$ are symmetric and approximate $G$ in the following sense:
  $$\lim_{n\to\infty} \int_{\T^2\times \T^2} (G_n-G)^2(x,y)\,\d x\d y =0, \quad \lim_{n\to\infty} \int_{\T^2} G_n(x,x)\,\d x =0.$$
Then the sequence of random variables $\<\omega\otimes \omega, G_n\>$ is a Cauchy sequence in mean square. We denote by $\<\omega\otimes \omega, G\>$ its limit.

Moreover, the limit is the same if $G_n$ is replaced by $\tilde G_n$ with the same properties and such that $\lim_{n\to\infty} \int_{\T^2\times \T^2} (G_n- \tilde G_n)^2(x,y)\,\d x\d y =0$.
\end{proposition}

\begin{proof}
The proofs are the same as those of \cite[Theorem 8]{F1}; we recall them here for completeness. Since $\lim_{n\to\infty} \int_{\T^2} G_n(x,x)\,\d x =0$, it is equivalent to show that $\<\omega\otimes \omega, G_n\>- \int_{\T^2} G_n(x,x)\,\d x$ is a Cauchy sequence in mean square. We have
  $$\aligned &\ \E\bigg[\Big|\<\omega\otimes \omega, G_n\>- \int_{\T^2} G_n(x,x)\,\d x- \<\omega\otimes \omega, G_m\>+ \int_{\T^2} G_m(x,x)\,\d x\Big|^2 \bigg]\\
  =&\ \E\bigg[\Big|\<\omega\otimes \omega, G_n-G_m\>- \int_{\T^2} (G_n-G_m)(x,x)\,\d x \Big|^2 \bigg]\\
  =&\ 2 \int_{\T^2\times \T^2} (G_n-G_m)^2(x,y)\,\d x\d y,
  \endaligned$$
where the last equality follows from (ii) and (iii) of Lemma \ref{lem-1}. This implies the Cauchy property, and thus $\<\omega\otimes \omega, G\>$ is well defined. The invariance property is proved similarly.
\end{proof}

Here is an example of the approximating functions $G_n$. Let $\chi: \T^2=[-1/2, 1/2]^2 \to [0,1]$ be a smooth and symmetric function with support in a small ball $B(0,r)$, and equal to 1 in $B(0,r/2)$. For any $n\geq 1$, set $\chi_n(x) = \chi(nx),\, x\in \T^2$. Define
  $$G_n(x) = \begin{cases}
  G(x)(1-\chi_n(x)), & x\neq 0;\\
  0, & x=0.
  \end{cases}$$
We regard $G_n$ as a function on $\T^2\times \T^2$ by setting $G_n(x,y) = G_n(x-y)$. Since $G_n(x,x) \equiv 0$, we have the following estimate (cf. Lemma \ref{lem-1}(iii)):
  \begin{equation}\label{eq-1}
  \E \big[(\<\omega \otimes \omega, G_n\>- \<\omega \otimes \omega, G\>)^2\big] \leq 2 \int_{\T^2\times \T^2} (G_n-G)^2(x,y)\,\d x\d y.
  \end{equation}

\subsection{Definition of the renormalized energy $\H$}

In this subsection we recall the definition of the renormalized energy $\H$ via the Galerkin approximation. To this end, let $\{e_k\}_{k\in \Z_0^2}$ be defined as
  \begin{equation}\label{real-basis}
  e_k(x) = \sqrt{2} \begin{cases}
  \cos(2\pi k\cdot x), & k\in \Z^2_+ , \\
  \sin(2\pi k\cdot x), & k\in \Z^2_-,
  \end{cases}
  \end{equation}
where $\Z_0^2= \Z^2 \setminus \{0\}$ and $\Z^2_+ = \big\{k\in \Z^2_0: (k_1 >0) \mbox{ or } (k_1=0,\, k_2>0) \big\}$ and $\Z^2_- = -\Z^2_+$. This family of functions is an orthonormal basis of square integrable functions on $\T^2$ with vanishing mean. Let $\omega$ be a white noise on $\T^2$, then the random series
  $$\omega = \sum_{k\in\Z_0^2} \<\omega, e_k\> e_k$$
converge in mean square in $H^{-1-\delta}(\T^2)$ for any $\delta>0$. For $N\geq 1$, define $\Lambda_N = \{k\in \Z_0^2: |k|\leq N\}$ and
  $$\bar\omega_N = \sum_{k\in \Lambda_N} \<\omega, e_k\> e_k, \quad u_N = K\ast \bar\omega_N,$$
where $K$ is the Biot--Savart kernel:
  $$K(x)= \nabla^\perp G(x) = -\frac{{\rm i}}{2\pi} \sum_{k\in \Z_0^2}  \frac{k^\perp}{|k|^2} {\rm e}^{2\pi {\rm i} k\cdot x},$$
with $\nabla^\perp = (\partial_2, -\partial_1)$ and $k^\perp = (k_2, -k_1)$. Set
  $$\mathcal E_N = \frac12 \int_{\T^2} |u_N(x)|^2\,\d x, \quad \tilde{\mathcal E}_N = \mathcal E_N - \E\mathcal E_N.$$
The following result is well known (see e.g. \cite[p. 593]{ARH} or \cite[Proposition 2.5]{AF2}).

\begin{proposition}\label{prop-2}
The sequence $\big\{\tilde{\mathcal E}_N \big\}_{N\geq 1}$ is Cauchy in $L^2(\Theta, \P)$. Denote its limit by $\H$ and call it the \emph{renormalized energy}; one has
  $$\H= \frac{1}{8\pi^2} \sum_{k\in \Z_0^2} \frac1{|k|^2}(\<\omega, e_k\>^2 -1)$$
and
  $$\E\big( |\H|^2 \big) = \frac{1}{32\pi^4} \sum_{k\in \Z_0^2} \frac1{|k|^4}.$$
\end{proposition}

\begin{proof}
We give the detailed computations since we need the exact coefficients. Note that
  $$K(x) =  \frac1{2\pi} \sum_{l\in \Z_0^2}  \frac{l^\perp}{|l|^2} \sin(2\pi l\cdot x).$$
Therefore, if $k\in \Z^2_+$, then
  $$\aligned (K\ast e_k)(x)&= \frac1{2\pi} \sum_{l\in \Z_0^2}  \frac{l^\perp}{|l|^2} \int_{\T^2} \sin(2\pi l\cdot (x-y)) \sqrt{2} \cos(2\pi k\cdot y)\,\d y\\
  &= \frac{\sqrt{2}}{2\pi} \frac{k^\perp}{|k|^2}\sin(2\pi k\cdot x) = - \frac{1}{2\pi} \frac{k^\perp}{|k|^2} e_{-k}(x),
  \endaligned$$
where the second equality is due to the fact that the integral vanishes unless $l= \pm k$. Similarly, if $k\in \Z^2_-$, then
  $$(K\ast e_k)(x) = \frac1{2\pi} \sum_{l\in \Z_0^2}  \frac{l^\perp}{|l|^2} \int_{\T^2} \sin(2\pi l\cdot (x-y)) \sqrt{2} \sin(2\pi k\cdot y)\,\d y = - \frac{1}{2\pi} \frac{k^\perp}{|k|^2} e_{-k}(x).$$
Thus,
  $$u_N = K\ast \bar\omega_N = - \frac{1}{2\pi} \sum_{k\in \Lambda_N} \frac{k^\perp}{|k|^2} \<\omega, e_k\> e_{-k}.$$
As a result,
  $$\mathcal E_N = \frac12 \int_{\T^2} |u_N(x)|^2\,\d x = \frac1{8\pi^2} \sum_{k,l\in \Lambda_N} \frac{k\cdot l}{|k|^2 |l|^2} \<\omega, e_k\> \<\omega, e_l\> \<e_{-k}, e_{-l}\> = \frac1{8\pi^2} \sum_{k\in \Lambda_N} \frac{1}{|k|^2} \<\omega, e_k\>^2. $$
Consequently,
  \begin{equation}\label{prop-2.1}
  \tilde{\mathcal E}_N = \frac1{8\pi^2} \sum_{k\in \Lambda_N} \frac{1}{|k|^2} (\<\omega, e_k\>^2 -1).
  \end{equation}

Next,
  $$\aligned \E \big[(\tilde{\mathcal E}_N)^2 \big] &= \frac1{64\pi^4} \sum_{k,l\in \Lambda_N} \frac{1}{|k|^2 |l|^2} \E \big[(\<\omega, e_k\>^2 -1) (\<\omega, e_l\>^2 -1) \big] \\
  &= \frac1{64\pi^4} \sum_{k\in \Lambda_N} \frac{1}{|k|^4} \E \big[\<\omega, e_k\>^4 -1 \big] = \frac1{32\pi^4} \sum_{k\in \Lambda_N} \frac{1}{|k|^4},
  \endaligned$$
since $\E\<\omega, e_k\>^4 =3$. The same calculations imply that $\big\{ \tilde{\mathcal E}_N \big\}_{N\geq 1}$ is a Cauchy sequence in $L^2(\Theta, \P)$ and the two desired equalities.
\end{proof}

As an application of the expression for the renormalized energy, we can prove

\begin{proposition}\label{prop-support}
The distribution of $\H$ is supported on the whole real line.
\end{proposition}

\begin{proof}
For any $a,b\in \R,\, a<b$, it suffices to show that $Z^{a,b}:= \P(\{\H\in [a,b]\})>0$. Without loss of generality, assume $b-a\leq 1$.

We define $\delta_0:= (b-a)/5$ and the remainder
  $$\mathcal R_N:= \frac1{8\pi^2} \sum_{|k|>N} \frac{1}{|k|^2} (\<\omega, e_k\>^2 -1).$$
Then $\H= \tilde{\mathcal E}_N + \mathcal R_N$ and, for all $N\geq 1$, the two r.v.'s $\tilde{\mathcal E}_N$ and $\mathcal R_N$ are independent of one another. Moreover,
  $$\{\H \in [a,b]\} \supset \big\{\tilde{\mathcal E}_N \in [a+\delta_0, b-\delta_0] \big\} \cap \{|\mathcal R_N| \leq \delta_0\},$$
therefore,
  $$\P\big(\{\H \in [a,b]\} \big) \geq \P\big( \big\{\tilde{\mathcal E}_N \in [a+\delta_0, b-\delta_0] \big\} \big)\, \P\big(\{|\mathcal R_N| \leq \delta_0\} \big).$$
Since $\mathcal R_N$ tends to 0 in the norm $L^2(\Theta,\P)$ as $N\to \infty$, we can find $N_0\in \Z_+$ such that $\P\big(\{|\mathcal R_N| \leq \delta_0\} \big) \geq 1/2$ for all $N\geq N_0$. Thus, it is enough to show that
  \begin{equation}\label{prop-support.1}
  \P\big( \big\{\tilde{\mathcal E}_{N_0} \in [a+\delta_0, b-\delta_0] \big\} \big)>0.
  \end{equation}

We define
  $$L= \frac1{8\pi^2} \sum_{k\in \Lambda_{N_0}} \frac{1}{|k|^2}$$
and consider three different cases according to the location of the origin 0 w.r.t. the middle subinterval $[a+2\delta_0, a+3\delta_0]$.

(i) $a+2\delta_0 >0$. Since $[a+3\delta_0, a+4\delta_0] \subset [a+\delta_0, b-\delta_0]$, it is sufficient to prove that
  \begin{equation}\label{prop-support.2}
  \P\big( \big\{\tilde{\mathcal E}_{N_0} \in [a+3\delta_0, a+4\delta_0] \big\} \big)>0.
  \end{equation}
Set $c_1:= (a+3\delta_0)/L$ and $c_2:= (a+4\delta_0)/L$ which are positive constants. Recall that $\{\<\omega,e_k\> \}_{k\in \Z_0^2}$ is a family of i.i.d. standard Gaussian r.v.'s; we have
  $$p_1 := \P\big( \big\{\<\omega,e_k\> ^2\in [1+c_1, 1+c_2] \big\} \big)>0.$$
The desired property \eqref{prop-support.2} follows from the next inclusion between events:
  $$\big\{\tilde{\mathcal E}_{N_0} \in [a+3\delta_0, a+4\delta_0] \big\} \supset \bigcap_{k\in \Lambda_{N_0}} \big\{\<\omega,e_k\> ^2\in [1+c_1, 1+c_2] \big\}.$$

(ii) $a+2\delta_0 \leq 0 \leq a+3\delta_0$. In this case, we have $[-\delta_0, \delta_0] \subset [a+\delta_0, b-\delta_0]$. Similar to case (i), we deduce the desired result from the two facts below:
  $$\big\{\tilde{\mathcal E}_{N_0} \in [-\delta_0, \delta_0] \big\} \supset \bigcap_{k\in \Lambda_{N_0}} \big\{\<\omega,e_k\> ^2\in [1-\delta_0/L, 1+\delta_0/L] \big\}$$
and
  $$p_2:= \P\big( \big\{\<\omega,e_k\> ^2\in [1-\delta_0/L, 1+\delta_0/L] \big\} \big)>0.$$

(iii) $a+3\delta_0< 0$. In this case, it suffices to show that
  \begin{equation}\label{prop-support.3}
  \P\big( \big\{\tilde{\mathcal E}_{N_0} \in [a+\delta_0, a+2\delta_0] \big\} \big)>0.
  \end{equation}
We assume $N_0$ is big enough such that the constant $L> -a$; then
  $$-1< c_3:= (a+\delta_0)/L< c_4:= (a+2\delta_0)/L<0.$$
We can get the inequality \eqref{prop-support.3} from the facts that
  $$\big\{\tilde{\mathcal E}_{N_0} \in [a+\delta_0, a+2\delta_0] \big\} \supset \bigcap_{k\in \Lambda_{N_0}} \big\{\<\omega,e_k\> ^2\in [1+c_3, 1+c_4] \big\}$$
and
  $$p_3:= \P\big( \big\{\<\omega,e_k\> ^2\in [1+c_3, 1+c_4] \big\} \big)>0.$$

Summarizing the above three cases, we complete the proof of \eqref{prop-support.1}.
\end{proof}

\subsection{The relation between $\<\omega\otimes \omega, G\>$ and $\H$}

In this part, for a white noise $\omega$, we follow the idea in \cite[Section 4.2]{FL-3} to show the relation between $\<\omega\otimes \omega, G\>$ and $\H$\,. Although we mainly work with real valued functions, we shall make use of the canonical complex orthonormal basis of $L^2\big(\T^2, \mathbb C\big)$: $\tilde e_k(x)= {\rm e}^{2\pi {\rm i} k\cdot x},\, k\in \Z^2, x\in \T^2$. Note that $\{\tilde e_k\otimes \tilde e_l \}_{k,l\in \Z^2}$ is an orthonormal basis of $L^2\big(\T^2\times \T^2, \mathbb C\big)$.

\begin{lemma}\label{lem-appendix}
Let $\omega$ be a white noise on $\T^2$. Assume $f\in C^\infty \big(\T^2\times \T^2,\R \big) $ is symmetric and $\int_{\T^2} f(x,x)\,\d x=0$. Then
  $$\<\omega\otimes \omega, f\> = \sum_{k,l\in \Z^2} f_{k,l} \<\omega, \tilde e_k\> \<\omega, \tilde e_l\> \quad \mbox{holds in } L^2(\Theta, \P),$$
where
  $$f_{k,l}= \<f, \tilde e_k\otimes \tilde e_l\>= \int_{\T^2\times \T^2} f(x,y)\tilde e_{k}(x) \tilde e_{l}(y)\,\d x\d y.$$
\end{lemma}

\begin{proof}
Denote by
  \begin{equation}\label{lem-appendix.0}
  \hat\Lambda_N = \{k\in \Z^2: |k|\leq N\} = \Lambda_N\cup \{0\}.
  \end{equation}
Since $f\in C^\infty(\T^2\times \T^2)$, the partial sum of the Fourier series
  $$f_N(x,y):= \sum_{k,l\in \hat\Lambda_N} f_{k,l}\, \tilde e_k(x) \tilde e_l(y) $$
converges to $f$, uniformly on $\T^2\times \T^2$ and in $L^2(\T^2\times \T^2)$. In particular,
  \begin{equation}\label{lem-appendix.1}
  \lim_{N\to\infty} \int_{\T^2} f_N(x,x)\,\d x = \int_{\T^2} f(x,x)\,\d x =0.
  \end{equation}
It is obvious that $f_N(x,y)$ is smooth and symmetric. By (ii) and (iii) in Lemma \ref{lem-1},
  \begin{equation*}
  \E \bigg[ \Big(\<\omega\otimes \omega, f- f_N\> + \int_{\T^2} f_N(x,x)\,\d x \Big)^2 \bigg] = 2 \int_{\T^2\times \T^2} (f- f_N)^2(x,y) \,\d x\d y.
  \end{equation*}
As a result,
  \begin{equation}\label{lem-appendix.2}
  \E \big[\<\omega\otimes \omega, f- f_N\>^2 \big] \leq 4\int_{\T^2\times \T^2} (f- f_N)^2(x,y) \,\d x\d y + 2 \bigg(\int_{\T^2} f_N(x,x)\,\d x \bigg)^2.
  \end{equation}

Next, note that
  $$\<\omega\otimes \omega, f_N\>= \sum_{k,l\in \hat\Lambda_N} f_{k,l} \<\omega, \tilde e_k\>\<\omega, \tilde e_l\>.$$
Therefore, by \eqref{lem-appendix.2},
  $$\aligned
  & \E \bigg[\Big(\<\omega\otimes \omega, f\> - \sum_{k,l\in \hat\Lambda_N} f_{k,l} \<\omega, \tilde e_k\>\<\omega, \tilde e_l\>\Big)^2 \bigg] \\
  \leq & \, 4\int_{\T^2\times \T^2} (f- f_N)^2(x,y) \,\d x\d y + 2 \bigg(\int_{\T^2} f_N(x,x)\,\d x \bigg)^2.
  \endaligned$$
Thanks to \eqref{lem-appendix.1}, the desired result follows by letting $N\to \infty$.
\end{proof}

We need the following simple equality.

\begin{lemma}\label{lem-appendix-2}
Let $\{a_{k,l} \}_{k,l\in \hat\Lambda_N} \subset \mathbb C$ be satisfying $a_{k,l}= a_{l,k}$, $\overline{a_{k,l}}= a_{-k,-l}$. Then
  \begin{equation*}
  \E \Bigg[ \bigg|\sum_{k,l\in \hat\Lambda_N} a_{k,l} \<\omega, \tilde e_k\>\<\omega, \tilde e_l\> - \sum_{k\in \hat\Lambda_N} a_{k,-k} \bigg|^2\Bigg] = 2\sum_{k,l\in \hat\Lambda_N} |a_{k,l}|^2.
  \end{equation*}
\end{lemma}

\begin{proof}
Since $\overline{\<\omega, \tilde e_k\>}= \<\omega, \tilde e_{-k}\>$, it is clear that $\sum_{k,l\in \hat\Lambda_N} a_{k,l} \<\omega, \tilde e_k\>\<\omega, \tilde e_l\>$ is real and
  \begin{equation}\label{lem-appendix-2.1}
  \sum_{k\in \hat\Lambda_N} a_{k,-k} = \E \Bigg(\sum_{k,l\in \hat\Lambda_N} a_{k,l} \<\omega, \tilde e_k\>\<\omega, \tilde e_l\> \Bigg).
  \end{equation}
We have
  $$\Bigg(\sum_{k,l\in \hat\Lambda_N} a_{k,l} \<\omega, \tilde e_k\>\<\omega, \tilde e_l\> \Bigg)^2 = \sum_{k,l,m,n \in \hat\Lambda_N} a_{k,l} a_{m,n} \<\omega, \tilde e_k\>\<\omega, \tilde e_l\> \<\omega, \tilde e_m\>\<\omega, \tilde e_n\>, $$
and by the Isserlis--Wick theorem,
  $$\aligned
  \E\big(\<\omega, \tilde e_k\>\<\omega, \tilde e_l\> \<\omega, \tilde e_m\>\<\omega, \tilde e_n\>\big) = &\  \E\big(\<\omega, \tilde e_k\>\<\omega, \tilde e_l\> \big) \E\big( \<\omega, \tilde e_m\>\<\omega, \tilde e_n\>\big)\\
  & + \E\big(\<\omega, \tilde e_k\>\<\omega, \tilde e_m\> \big) \E\big( \<\omega, \tilde e_l\>\<\omega, \tilde e_n\>\big) \\
  & + \E\big(\<\omega, \tilde e_k\>\<\omega, \tilde e_n\> \big) \E\big( \<\omega, \tilde e_l\>\<\omega, \tilde e_m\>\big) \\
  =&\ \delta_{k,-l}\delta_{m,-n} + \delta_{k,-m}\delta_{l,-n} + \delta_{k,-n}\delta_{l,-m}.
  \endaligned$$
Therefore,
  $$\aligned \E \Bigg(\sum_{k,l\in \hat\Lambda_N} a_{k,l} \<\omega, \tilde e_k\>\<\omega, \tilde e_l\> \Bigg)^2 &= \sum_{k,m \in \hat\Lambda_N} a_{k,-k} a_{m,-m} + \sum_{k,l \in \hat\Lambda_N} a_{k,l} a_{-k,-l}  + \sum_{k,l \in \hat\Lambda_N} a_{k,l} a_{-l,-k}  \\
  &= \Bigg(\sum_{k \in \hat\Lambda_N} a_{k,-k} \Bigg)^2 + 2 \sum_{k,l \in \hat\Lambda_N} |a_{k,l}|^2,
  \endaligned$$
where we have used the facts $a_{-l,-k} = a_{-k,-l} = \overline{a_{k,l}}$. Combining this equality with \eqref{lem-appendix-2.1} finishes the proof.
\end{proof}

Recall the definition of $\hat\Lambda_N$ in \eqref{lem-appendix.0}. To simplify the notations, we introduce
  $$\hat \omega_N = \hat \Pi_N \omega = \sum_{k\in \hat\Lambda_N} \<\omega, \tilde e_k\> \tilde e_k.$$
Then
  \begin{equation}\label{partial-sum}
  \<\hat \omega_N\otimes \hat \omega_N, G\> = \sum_{k,l\in \hat\Lambda_N} \<G, \tilde e_k\otimes \tilde e_l\> \<\omega, \tilde e_k\>\<\omega, \tilde e_l\>
  \end{equation}
is the partial sum of the series.

\begin{lemma}\label{lem-2}
We have
  $$\<\hat \omega_N\otimes \hat \omega_N, G\> = -\frac1{4\pi^2} \sum_{k\in \Lambda_N} \frac1{|k|^2} \<\omega, e_k\>^2.$$
\end{lemma}

\begin{proof}
Recall that
  $$G(x) = -\frac1{4\pi^2} \sum_{k\in \Z_0^2} \frac1{|k|^2} {\rm e}^{2\pi {\rm i} k\cdot x} = -\frac1{4\pi^2} \sum_{k\in \Z_0^2} \frac1{|k|^2} \tilde e_k(x). $$
Therefore, for $l\neq 0$,
  $$ (G\ast \tilde e_l)(x) = -\frac1{4\pi^2}  \frac1{|l|^2} \tilde e_l(x),$$
which implies that
  $$\<G, \tilde e_k\otimes \tilde e_l\> = \int_{\T^2} \tilde e_k(x) (G\ast \tilde e_l)(x)\,\d x= -\frac1{4\pi^2}  \frac1{|l|^2} \delta_{k,-l}. $$
Hence,
  $$\<\hat \omega_N\otimes \hat \omega_N, G\> = -\frac1{4\pi^2} \sum_{k\in \Lambda_N} \frac1{|k|^2} \<\omega, \tilde e_k\>\<\omega, \tilde e_{-k}\> = -\frac1{4\pi^2} \sum_{k\in \Lambda_N} \frac1{|k|^2} |\<\omega, \tilde e_k\>|^2.$$
The desired identity follows from $|\<\omega, \tilde e_k\>|^2 = \frac12 (\<\omega, e_k\>^2 + \<\omega, e_{-k}\>^2)$ for all $k\in \Lambda_N$.
\end{proof}

Now we can prove the main result of this section.

\begin{theorem}\label{thm-coincidence}
Let $\omega$ be a white noise on $\T^2$. Almost surely, it holds
  $$\<\omega\otimes \omega, G\> =-2\, \H\, .$$
\end{theorem}

\begin{proof}
Let $G_n$ be the smooth functions defined at the end of Section \ref{sec-2.1}.  We have
  \begin{equation}\label{prop-appendix.1}
  \aligned
  &\, \E \big[\big(\<\omega\otimes \omega, G\> + 2\, \H \big)^2 \big] \\
  \leq & \, 4\, \E \big[\<\omega\otimes \omega, G - G_n\>^2 \big] + 4\, \E\big[\big(\<\omega\otimes \omega, G_n\> - \<\hat \omega_N\otimes \hat \omega_N, G_n\> \big)^2 \big] \\
  & + 4\, \E \big[ \big( \<\hat \omega_N\otimes \hat \omega_N, G_n\> + 2\, \tilde{\mathcal E}_N \big)^2 \big] + 16\, \E \big[\big( \tilde{\mathcal E}_N - \H \big)^2 \big].
  \endaligned
  \end{equation}

We deal with these terms one-by-one. By \eqref{eq-1},
  \begin{equation}\label{prop-appendix.2}
  \E \big[\<\omega\otimes \omega, G-G_n\>^2 \big] \leq 2 \int_{\T^2\times \T^2} ( G-G_n)^2(x,y) \,\d x\d y.
  \end{equation}
Next, for any fixed $n\geq 1$, Lemma \ref{lem-appendix} implies
  \begin{equation}\label{prop-appendix.2.5}
  \E \big[ \big(\<\hat \omega_N\otimes \hat \omega_N, G_n\> - \<\omega\otimes \omega, G_n\>\big)^2\big]\to 0 \quad \mbox{as } N\to \infty.
  \end{equation}
Moreover, by Proposition \ref{prop-2}, the last term in \eqref{prop-appendix.1} vanishes as $N\to \infty$.

It remains to treat the third term on the r.h.s. of \eqref{prop-appendix.1}. By \eqref{prop-2.1} and Lemma \ref{lem-2},
  \begin{equation*}
  -2\, \tilde{\mathcal E}_N= \<\hat \omega_N\otimes \hat \omega_N, G\> - \E \<\hat \omega_N\otimes \hat \omega_N, G\>.
  \end{equation*}
Therefore,
  $$\aligned
  &\, \E \big[ \big( \<\hat \omega_N\otimes \hat \omega_N, G_n\> + 2\, \tilde{\mathcal E}_N \big)^2 \big] \\
  = & \, \E \big[ \big( \<\hat \omega_N\otimes \hat \omega_N, G_n-G\> - \E\<\hat \omega_N\otimes \hat \omega_N, G_n-G\> + \E \<\hat \omega_N\otimes \hat \omega_N, G_n\> \big)^2 \big] \\
  \leq &\, 2\, \E \big[ \big( \<\hat \omega_N\otimes \hat \omega_N, G_n-G\> - \E\<\hat \omega_N\otimes \hat \omega_N, G_n-G\> \big)^2 \big] + 2 \big[\E \<\hat \omega_N\otimes \hat \omega_N, G_n\> \big]^2.
  \endaligned$$
By \eqref{partial-sum} and Lemma \ref{lem-appendix-2},
  $$\aligned
  &\, \E \big[\big(\<\hat \omega_N\otimes \hat \omega_N, G_n- G \> - \E \<\hat \omega_N\otimes \hat \omega_N, G_n- G \> \big)^2 \big]\\
  = &\, 2 \sum_{k,l\in \hat\Lambda_N} \big| \big\<G_n- G, \tilde e_k\otimes \tilde e_l \big\> \big|^2 \leq 2 \int_{\T^2\times \T^2} (G_n- G)^2(x,y)\,\d x\d y .
  \endaligned $$
Hence,
  $$\aligned
  \E \big[ \big( \<\hat \omega_N\otimes \hat \omega_N, G_n\> + 2\tilde{\mathcal E}_N \big)^2 \big] \leq &\, 4 \int_{\T^2\times \T^2} (G_n- G)^2(x,y)\,\d x\d y + 2 \big[ \E \<\hat \omega_N\otimes \hat \omega_N, G_n \> \big]^2.
  \endaligned $$
As a result of \eqref{prop-appendix.2.5},
  \begin{equation*}
  \lim_{N\to \infty} \E \<\hat \omega_N\otimes \hat \omega_N, G_n\> =  \E \<\omega\otimes \omega, G_n\> = \int_{\T^2} G_n(x,x)\,\d x=0,
  \end{equation*}
where the second step is due to Lemma \ref{lem-1}(ii). Thus,
  $$\limsup_{N\to\infty} \E \big[ \big( \<\hat \omega_N\otimes \hat \omega_N, G_n\> + 2\, \tilde{\mathcal E}_N \big)^2 \big] \leq 4 \int_{\T^2\times \T^2} (G_n- G)^2(x,y)\,\d x\d y.$$
Combining the above inequality with \eqref{prop-appendix.1}--\eqref{prop-appendix.2.5}, letting $N\to \infty$ in \eqref{prop-appendix.1} yield
  $$\limsup_{N\to\infty} \E \big[\big(\<\omega\otimes \omega, G\> + 2\, \H \big)^2 \big] \leq 24 \int_{\T^2\times \T^2} (G_n- G)^2(x,y)\,\d x\d y.$$
We finish the proof by sending $n\to \infty$.
\end{proof}

\section{Proof of the main result}

In this section, we first show that the Hamiltonian $\mathcal H_N$ converge weakly to the renormalized energy $\H$, by making use of the weak convergence of the random point vortices to the white noise. Thanks to the fact that the distribution of $\H$ has a density, finally we are able to prove Theorem \ref{main-thm}.

First of all, we prove the following intermediate result.

\begin{proposition}\label{prop-weak}
The Hamiltonian $\mathcal H_N$ defined in \eqref{Hamiltonian} converge weakly to the renormalized energy $\H$.
\end{proposition}

By \eqref{Hamiltonian-1} and Theorem \ref{thm-coincidence}, it suffices to prove that $\<\omega_N \otimes \omega_N, G\>$ converge weakly to $\<\omega \otimes \omega, G\>$, where $\omega_N$ is the random point vortices and $\omega$ is a white noise on $\T^2$. This result seems to be obvious, thanks to the weak convergence of $\omega_N$ to $\omega$, see \cite[Proposition 21]{F1}. However, in view of the proof of our main result, we give the details here.

The following equality will be very useful in the sequel (cf. \cite[Lemma 23]{F1}): if $f\in L^2(\T^2\times \T^2,\R)$ is symmetric and $f(x,x)\equiv 0$, then
  \begin{equation}\label{point-vortices}
  \E\big[ \<\omega_N \otimes \omega_N, f\>^2 \big] = 2 \int_{\T^2\times \T^2} f(x,y)^2\,\d x\d y.
  \end{equation}
Note that, we do not need the boundedness of $f$ for proving the above equality and, by convention, $G(x,x)=0$ for all $x\in \T^2$.

Recall that $\mu_N$ is the law of $\omega_N$ on $H^{-1-}(\T^2)$ and that the sequence $\{\mu_N\}_{N\geq 1}$ converge weakly to the enstrophy measure $\mu$. Thus by the Skorokhod theorem, there exists a new probability space $\big( \tilde\Theta, \tilde{\mathcal F}, \tilde\P \big)$ and a sequence of random variables $\tilde\omega_N: \tilde\Theta \to H^{-1-}(\T^2)$, and $\tilde\omega: \tilde\Theta \to H^{-1-}(\T^2)$ such that
\begin{itemize}
\item[(a)] $\tilde\omega_N \stackrel{\L}{\sim} \mu_N$ and $\tilde\omega \stackrel{\L}{\sim} \mu$;
\item[(b)] $\tilde\P$-a.s., $\tilde\omega_N $ converge  to $\tilde\omega$ as $N\to \infty$.
\end{itemize}
In particular, $\tilde\omega$ is a white noise on $\T^2$.

Let $f: \R \to \R$ be any bounded and uniformly continuous function. We need to show
  $$\lim_{N\to \infty} \E f\big(\<\omega_N \otimes \omega_N, G\>\big) = \E f\big(\<\omega \otimes \omega, G\>\big). $$
We have
  $$\E f\big(\<\omega_N \otimes \omega_N, G\>\big) - \E f\big(\<\omega \otimes \omega, G\>\big) = \tilde \E f\big(\<\tilde\omega_N \otimes \tilde\omega_N, G\>\big) - \tilde \E f\big(\<\tilde\omega \otimes \tilde\omega, G\>\big).$$
From the next result we deduce that the above quantity vanishes as $N\to \infty$.

\begin{lemma}\label{lem-3-1}
We have
  $$\lim_{N\to \infty} \tilde \E \big|\<\tilde\omega_N \otimes \tilde\omega_N, G\>- \<\tilde\omega \otimes \tilde\omega, G\>\big| =0.$$
\end{lemma}

\begin{proof}
Denote the expectation by $I_N$. Let $G_n$ be the approximating functions given at the end of Section \ref{sec-2.1}. By the triangle inequality,
  \begin{equation}\label{sec-weak-1}
  \aligned I_N \leq &\, \tilde \E \big|\<\tilde\omega_N \otimes \tilde\omega_N, G\> - \<\tilde\omega_N \otimes \tilde\omega_N, G_n\> \big| + \tilde \E \big| \<\tilde\omega_N \otimes \tilde\omega_N, G_n\> - \<\tilde\omega \otimes \tilde\omega, G_n\> \big| \\
  &+ \tilde \E \big| \<\tilde\omega \otimes \tilde\omega, G_n\> - \<\tilde\omega \otimes \tilde\omega, G\> \big|.
  \endaligned
  \end{equation}
We denote the three terms by $I_{N,i},\, i=1,2,3$. Cauchy's inequality yields
  $$I_{N,1} \leq \Big(\tilde \E \big|\<\tilde\omega_N \otimes \tilde\omega_N, G\> - \<\tilde\omega_N \otimes \tilde\omega_N, G_n\> \big|^2 \Big)^{1/2} = \bigg(2\int_{\T^2\times \T^2} (G-G_n)^2(x,y) \,\d x\d y\bigg)^{1/2},$$
where in the second step we have used \eqref{point-vortices}. Similarly, by \eqref{eq-1},
  $$I_{N,3} \leq \bigg(2\int_{\T^2\times \T^2} (G-G_n)^2(x,y) \,\d x\d y\bigg)^{1/2}.$$
Next, for any fixed $n\geq 1$, the family $\big\{ \<\tilde\omega_N \otimes \tilde\omega_N, G_n\> \big\}_{N\geq 1}$ is bounded in $L^2\big(\tilde\P \big)$ by \eqref{point-vortices}, hence it is uniformly integrable. Moreover, $\tilde\P$-a.s.,
  $$\<\tilde\omega_N \otimes \tilde\omega_N, G_n\> \to \<\tilde\omega \otimes \tilde\omega, G_n\> \quad \mbox{as } N\to \infty, $$
due to the a.s. convergence of $\tilde\omega_N$ to $\tilde\omega$. Therefore,
  $$\lim_{N\to \infty} I_{N,2} =0.$$
Summarizing the above discussions, we first let $N\to \infty$ and then $n\to \infty$ in \eqref{sec-weak-1} to deduce that $\lim_{N\to \infty} I_N =0$.
\end{proof}

As a consequence, we can prove

\begin{corollary}\label{3-cor}
For any nontrivial interval $[a,b]$, one has
  \begin{equation}\label{eq-2}
  \lim_{N\to \infty} \P(\{\mathcal H_N \in [a,b]\}) = \P(\{\H \in [a,b]\}).
  \end{equation}
\end{corollary}

\begin{proof}
By \cite[Theorem 8.3]{Malliavin} (see also \cite[Theorems 3.2 and 3.3]{Cipriano}), the renormalized energy $\H$ is infinitely differentiable in the sense of Malliavin and it is non-degenerate, which implies that, as a real valued random variable, the distribution $\nu$ of $\H$ has a density w.r.t. the Lebesgue measure on $\R$. Thus any interval $[a,b]$ is a $\nu$-continuous set, that is, the boundary of $[a,b]$ (i.e. $\{a,b\}$) is $\nu$-negligible. On the other hand, Proposition \ref{prop-weak} tells us that the distributions on $\R$ of $\mathcal H_N$ converge weakly to $\nu$ as $N\to\infty$. Therefore, the desired limit holds true.
\end{proof}

Finally we are ready to prove the main result.

\begin{proof}[Proof of Theorem \ref{main-thm}]
Taking into account Proposition \ref{prop-support} and Corollary \ref{3-cor}, it is sufficient to show that, for any bounded and uniformly continuous function $F:H^{-1-}(\T^2) \to \R$, one has
  $$\lim_{N\to \infty} \E \big[ F(\omega_N) {\bf 1}_{[a,b]}(\mathcal H_N) \big]= \E \big[ F(\omega) {\bf 1}_{[a,b]}(\H) \big],$$
where $\omega_N$ and $\omega$ denote the random point vortices and the white noise respectively.

We follow the idea of the last subsection and use the Skorokhod theorem. Then, adopting the notations given there,
  $$\E \big[ F(\omega_N) {\bf 1}_{[a,b]}(\mathcal H_N) \big] - \E \big[ F(\omega) {\bf 1}_{[a,b]}(\H) \big] = \tilde\E \big[ F(\tilde\omega_N) {\bf 1}_{[a,b]}(\tilde{\mathcal H}_N) \big] - \tilde\E \big[ F(\tilde\omega) {\bf 1}_{[a,b]}(\tilde{\H}) \big], $$
where the notations with a tilde denote quantities on the new probability space $\big( \tilde\Theta, \tilde{\mathcal F}, \tilde\P \big)$. Denote by $J_N$ the difference on the right hand side; then,
  $$|J_N| \leq \tilde\E |F(\tilde\omega_N) - F(\tilde\omega)| + \|F\|_\infty \tilde\E \big|{\bf 1}_{[a,b]}(\tilde{\mathcal H}_N) - {\bf 1}_{[a,b]}(\tilde{\H}) \big|.$$
The first term tends to zero by the dominated convergence theorem and the $\tilde\P$-a.s. convergence of $\tilde \omega_N$ to $\tilde \omega$. To show that the second one also vanishes as $N\to \infty$, we take a sequence of bounded continuous functions such that:
  $$ f_n(t) = \begin{cases}
  1, & t\in [a, b]; \\
  0, & t\in (-\infty, a-1/n ] \cup [b+ 1/n, +\infty ); \\
  \mbox{linear function}, & t\in [a-1/n, a ] \cup [b, b+ 1/n].
  \end{cases} $$
Then,
  \begin{equation}\label{proof-1}
  \aligned &\, \tilde\E \big|{\bf 1}_{[a,b]}(\tilde{\mathcal H}_N) - {\bf 1}_{[a,b]}(\tilde{\H}) \big| \\
  \leq &\, \tilde\E \big|{\bf 1}_{[a,b]}(\tilde{\mathcal H}_N) - f_n(\tilde{\mathcal H}_N)\big| + \tilde\E \big|f_n(\tilde{\mathcal H}_N) - f_n(\tilde{\H}) \big| + \tilde\E \big| f_n(\tilde{\H}) - {\bf 1}_{[a,b]}(\tilde{\H}) \big|.
  \endaligned
  \end{equation}
By Lemma \ref{lem-3-1}, we know that $\tilde{\mathcal H}_N = -\frac12 \<\tilde\omega_N \otimes \tilde\omega_N, G\>$ converge in $L^1\big(\tilde\P \big)$ to $\tilde{\H} = -\frac12 \<\tilde\omega \otimes \tilde\omega, G\>$. For fixed $n\in \N$, $f_n$ is Lipschitz continuous with $\|f_n\|_{\rm Lip} =n$; therefore,
  $$\lim_{N\to \infty} \tilde\E \big|f_n(\tilde{\mathcal H}_N) - f_n(\tilde{\H}) \big|=0.$$
Next, let $\nu_N$ be the distribution of $\mathcal H_N$, and thus also of $\tilde{\mathcal H}_N$. We have
  $$\tilde\E \big|{\bf 1}_{[a,b]}(\tilde{\mathcal H}_N) - f_n(\tilde{\mathcal H}_N)\big| \leq \nu_N\big([a-1/n, a ] \cup [b, b+ 1/n ] \big),$$
hence, by Proposition \ref{prop-weak},
  $$\limsup_{N\to \infty} \E \big|{\bf 1}_{[a,b]}(\tilde{\mathcal H}_N) - f_n(\tilde{\mathcal H}_N)\big| \leq \nu\big([a-1/n, a ] \cup [b, b+ 1/n ] \big),$$
where $\nu$ is the distribution of $\H$ which is the same as that of $\tilde{\H}$. Finally,
  $$\tilde\E \big|f_n(\tilde{\H}) - {\bf 1}_{[a,b]}(\tilde{\H}) \big| \leq \nu\big([a-1/n, a ] \cup [b, b+ 1/n ] \big).$$
Recall that $\nu$ is absolutely continuous w.r.t. the Lebesgue measure. Therefore, first letting $N\to \infty$ and then $n\to \infty$ in \eqref{proof-1}, we complete the proof.
\end{proof}

\section{Triviality of cluster points}

In this part, following the discussions at the end of \cite[Section 5.4]{Lions} (see also \cite{Neri}), we study the limit behavior of the correlation functions (i.e. marginal distributions) of the energy conditional measures $\lambda_N^{a,b}$ on the ``flat space'' $(\R\times \T^2)^N$. Here, for $a,b\in \R,\, a<b$,
  \begin{equation}\label{flat-cond-measure}
  \lambda_N^{a,b} = \frac1{Z_N^{a,b}} {\bf 1}_{\{\mathcal H_N((\xi_1,x_1), \ldots, (\xi_N,x_N))\in [a,b]\}} \lambda_N,
  \end{equation}
where $\lambda_N$ is defined in \eqref{invariant-meas} and $Z_N^{a,b}$ is the normalizing constant:
  $$Z_N^{a,b}= \int_{(\R\times \T^2)^N} {\bf 1}_{\{\mathcal H_N((\xi_1,x_1), \ldots, (\xi_N,x_N))\in [a,b] \}} \,\d \lambda_N =\P(\{\mathcal H_N \in [a,b]\}).$$
Note that the measure $\mu_N^{a,b}$ defined in the introduction is the image measure of $\lambda_N^{a,b}$ under the map $\mathcal T_N: (\R\times \T^2)^N \to H^{-1-}(\T^2) $ defined as
  \begin{equation}\label{map}
  \big((\xi_1, x_1), \ldots, (\xi_N, x_N)\big) \stackrel{\mathcal T_N}{\longrightarrow} \frac1{\sqrt N} \sum_{i=1}^N \xi_i \delta_{x_i}.
  \end{equation}

We introduce the notations $\tilde x_i=(\xi_i, x_i)\in \R\times \T^2 $ and $X_N= (\tilde x_1, \ldots, \tilde x_N)$. Denote by $\d \tilde x_i = \d x_i \mathcal N(\d \xi_i)$ and $\d X_N= \d \tilde x_1 \ldots \d \tilde x_N$. Let
  $$\rho^N(X_N)= {\bf 1}_{\{ \mathcal H_N(X_N)\in [a,b] \}} /Z_N^{a,b}$$
be the density function of the conditional probability measure $\lambda_N^{a,b}$ on $(\R\times \T^2)^N$. The \emph{correlation functions} $\rho^N_j\, (1\leq j\leq N)$ are defined as follows: $\rho^N_N= \rho^N$ and for $1\leq j\leq N-1$,
  $$\rho^N_j(\tilde x_1, \ldots, \tilde x_j)= \int_{(\R\times \T^2)^{N-j}} \rho^N(X_N) \,\d \tilde x_{j+1} \ldots \d \tilde x_N. $$
Each function is a probability density (for the first $j$ point vortices) and is symmetric in $(\tilde x_1, \ldots, \tilde x_j)$, thanks to the symmetry of $\rho^N$ in $(\tilde x_1, \ldots, \tilde x_N)$. To simplify the notations, we introduce
  $$X_j= (\tilde x_1, \ldots, \tilde x_j), \quad X^{N-j}= (\tilde x_{j+1}, \ldots, \tilde x_N), \quad 1\leq j\leq N-1.$$

First of all, we have the following simple result (see \cite[(22)]{Lions} or \cite[Proposition 6]{Neri}).

\begin{lemma}\label{lem-4-1}
For any $1\leq j\leq N-1$,
  $$\int_{(\R\times \T^2)^N} \rho^N \log \rho^N\,\d X_N \geq \int_{(\R\times \T^2)^j} \rho^N_j \log \rho^N_j\,\d X_j + \int_{(\R\times \T^2)^{N-j}} \rho^N_{N-j} \log \rho^N_{N-j}\,\d X^{N-j}.$$
\end{lemma}

\begin{proof}
We include the proof for the reader's convenience. It is well known that $t\log t \geq t -1$ for all $t\geq 0$. Therefore,
  $$\rho^N(X_N) \log\bigg(\frac{\rho^N(X_N)}{\rho^N_j(X_j) \rho^N_{N-j}(X^{N-j})} \bigg) + \rho^N_j(X_j) \rho^N_{N-j}(X^{N-j}) -\rho^N(X_N) \geq 0,$$
which implies
  $$\int_{(\R\times \T^2)^N} \rho^N(X_N) \log\bigg(\frac{\rho^N(X_N)}{\rho^N_j(X_j) \rho^N_{N-j}(X^{N-j})} \bigg) \,\d X_N \geq 0.$$
This is equivalent to the desired inequality.
\end{proof}

Next, we prove

\begin{proposition}\label{prop-3}
For any fixed $j\geq 1$, $\big\{\rho^N_j \big\}_{N\geq j}$ is weakly compact in $L^1\big( (\R\times \T^2)^j, \d X_j \big)$.
\end{proposition}

\begin{proof}
For any $N\geq j$, there exist $m= m(j,N)\in \mathbb N$ and $r=r(j,N)\in \{0,1,\ldots, j-1\}$ such that $N= mj+r$. By Lemma \ref{lem-4-1},
  $$\int_{(\R\times \T^2)^N} \rho^N \log \rho^N\,\d X_N \geq m \int_{(\R\times \T^2)^j} \rho^N_j \log \rho^N_j\,\d X_j + r \int_{\R\times \T^2} \rho^N_1 \log \rho^N_1\,\d X_1.$$
Using the inequality $t\log t \geq t -1\ (t\geq 0)$, it is clear that $\int_{\R\times \T^2} \rho^N_1 \log \rho^N_1\,\d X_1\geq 0$. Thus
  $$\int_{(\R\times \T^2)^j} \rho^N_j \log \rho^N_j\,\d X_j \leq \frac1m \int_{(\R\times \T^2)^N} \rho^N \log \rho^N\,\d X_N = \frac1m \log \frac1{Z_N^{a,b}},$$
where the last step follows from the definition of $\rho^N$. Noticing that $\frac1m = O\big(\frac1N \big)$, thus by \eqref{eq-2}, the right hand side vanishes as $N\to \infty$. In particular, we conclude that $\big\{\rho^N_j \log \rho^N_j \big\}_{N\geq j}$ is bounded in $L^1\big( (\R\times \T^2)^j, \d X_j \big)$. The proof is complete.
\end{proof}

We say that a family $\{\rho_j\}_{j\geq 1}$ of probability densities is a weak cluster point of $\big\{\rho^N_j \big\}_{j\geq 1}$ if there exists a subsequence $\{N_k\}_{k\geq 1}$ of integers, such that, for any $j\geq 1$, $\rho^{N_k}_j$ converge weakly to $\rho_j$ in $L^1\big( (\R\times \T^2)^j, \d X_j \big)$. Now we prove the main result of this section.

\begin{theorem}\label{thm-triviality}
Any weak cluster point $\{\rho_j\}_{j\geq 1}$ of $\big\{\rho^N_j \big\}_{j\geq 1}$ is trivial, that is, for any $j\geq 1$, $\rho_j = 1$ almost surely on $(\R\times \T^2)^j$. Consequently, for any $j\geq 1$, the whole sequence $\big\{\rho^N_j \big\}_{N\geq j}$ converge weakly to 1.
\end{theorem}

\begin{proof}
Fix any $\eps>0$ and $j\geq 1$; let
  $$\mathcal C_\eps = \bigg\{ u\in L^1\big( (\R\times \T^2)^j, \d X_j \big): u \geq 0, \int_{(\R\times \T^2)^j} u\log u\,\d X_j \leq \eps \bigg\}.$$
Let $\{N_k\}_{k\geq 1}$ be the subsequence such that $\rho^{N_k}_j$ converge weakly to $\rho_j$ in $L^1\big( (\R\times \T^2)^j, \d X_j \big)$. By the proof of Proposition \ref{prop-3}, we have $\rho^{N_k}_j \in \mathcal C_\eps$ for all $k$ big enough. Therefore, $\rho_j$ is a weak cluster point of $\mathcal C_\eps$, which is a convex subset of $L^1\big( (\R\times \T^2)^j, \d X_j \big)$. Since the weak closure of $\mathcal C_\eps$ coincides with the strong one, there exists a sequence of functions $\{u_n \}\subset \mathcal C_\eps$ which converge strongly to $\rho_j$ in $L^1\big( (\R\times \T^2)^j, \d X_j \big)$. Along a subsequence, $u_n$ converge to $\rho_j$ almost everywhere. Therefore, by Fatou's lemma, we have
  $$\int_{(\R\times \T^2)^j} \rho_j\log \rho_j\, \d X_j \leq \liminf_{n\to\infty} \int_{(\R\times \T^2)^j} u_n\log u_n\, \d X_j \leq \eps.$$
The arbitrariness of $\eps>0$ leads to $\int_{(\R\times \T^2)^j} \rho_j\log \rho_j\, \d X_j =0$, which implies $\rho_j = 1$ almost surely. The last assertion follows from the weak compactness of $\big\{\rho^N_j \big\}_{N\geq j}$ and the uniqueness of the weak limit.
\end{proof}

Theorem \ref{thm-triviality} is a propagation-of-chaos type result, which means, in the limit, the joint distribution $\rho_j$ of $j$ point vortices is the $j$-th power of that of one point vortex. The weak cluster point obtained above gives a trivial solution to the mean field equation
  $$\rho(\xi,x)= \frac1{Z_\beta} e^{-\beta \xi U_\rho (x)},\quad \beta\in \R,$$
where $Z_\beta$ is the normalizing constant and $U_\rho$ is the averaged stream function:
  $$U_\rho (x)= \int_{\R\times \T^2} \xi G(x,y) \rho(\xi, y)\, \mathcal N(\d\xi) \d y,\quad x\in \T^2.$$
In our case, $\rho=1$ a.s. and $U_\rho =0$ a.s. The corresponding free energy $F(1)= S(1) + \beta E(1)=0$, where the entropy
  $$S(\rho)= \int_{\R\times \T^2} \rho(\tilde x)\log \rho(\tilde x) \,\d\tilde x$$
and the energy
  $$E(\rho)= \int_{(\R\times \T^2)^2} \mathcal H_2(\tilde x_1, \tilde x_2) \rho(\tilde x_1)\rho(\tilde x_2)\, \d \tilde x_1\d \tilde x_2. $$

We conclude this section by showing that, under the measure $\lambda_N^{a,b}= \rho^N(X_N)\,\d X_N$, the empirical measure $\frac1N \sum_{i=1}^N \delta_{\tilde x_i}$ converges weakly to the trivial measure $\d\tilde x = \mathcal N(\d\xi) \d x$ on $\R\times\T^2$.

\begin{corollary}
For any $\phi\in C_b(\R\times \T^2)$,
  $$\lim_{N\to\infty} \int_{(\R\times \T^2)^N} \bigg|\frac1N \sum_{i=1}^N \phi(\tilde x_i) - \int_{\R\times \T^2} \phi(\tilde x) \,\d\tilde x\bigg|^2 \,\d \lambda_N^{a,b}=0.$$
\end{corollary}

\begin{proof}
We denote by $I_N$ the integral. Expanding the square in the integral leads to
  $$\aligned
  I_N = &\, \frac1{N^2} \sum_{i,j=1}^N \int_{(\R\times \T^2)^N} \phi(\tilde x_i) \phi(\tilde x_j) \,\d \lambda_N^{a,b} +\bigg(\int_{\R\times \T^2} \phi(\tilde x) \,\d\tilde x \bigg)^2 \\
  & - \frac2N \bigg(\int_{\R\times \T^2} \phi(\tilde x) \,\d\tilde x \bigg) \sum_{i=1}^N \int_{(\R\times \T^2)^N} \phi(\tilde x_i) \,\d \lambda_N^{a,b}.
  \endaligned$$
Note that $\lambda_N^{a,b}= \rho^N(X_N)\,\d X_N$. Using the marginal densities $\rho^N_j,\, j=1,2$, we have
  $$\aligned
  I_N =&\, \frac{N-1}{N} \int_{(\R\times \T^2)^2} \phi(\tilde x_1) \phi(\tilde x_2) \rho^N_2(\tilde x_1, \tilde x_2)\,\d \tilde x_1\d \tilde x_2 + \frac{1}{N} \int_{\R\times \T^2} \phi(\tilde x_1)^2 \rho^N_1(\tilde x_1)\,\d \tilde x_1 \\
  &\, + \bigg(\int_{\R\times \T^2} \phi(\tilde x) \,\d\tilde x \bigg)^2 - 2 \bigg(\int_{\R\times \T^2} \phi(\tilde x) \,\d\tilde x \bigg) \int_{\R\times \T^2} \phi(\tilde x_1) \rho^N_1(\tilde x_1)\,\d \tilde x_1 .
  \endaligned$$
Now we finish the proof by letting $N\to\infty$ and using Theorem \ref{thm-triviality}.
\end{proof}

\section{Energy conditional solutions to 2D Euler equations}

In this part we show the existence of solutions to 2D Euler equations whose renormalized energy is confined in an interval $[a,b]$. First we give the precise meaning of the solution.

\begin{definition}\label{def}
Let $a,b\in \R, a<b$ be fixed. A stochastic process $\omega_\cdot$ defined on some probability space $(\Theta, \mathcal F, \P)$ with trajectories in $C\big( [0,T], H^{-1-}(\T^2)\big)$ is called an energy conditional solution of the 2D Euler equations if for any $t\in [0,T]$, $\omega_t$ has the distribution $\mu^{a,b}$ defined in \eqref{energy-cond-meas}, and for any $\phi\in C^\infty(\T^2)$, $\P$-a.s.,
   \begin{equation}\label{weak-vorticity}
   \<\omega_t, \phi\>= \<\omega_0, \phi\> + \int_0^t \<\omega_s\otimes \omega_s, H_\phi\>\,\d s \quad \mbox{for all } t\in [0,T].
   \end{equation}
\end{definition}

The above equation is called the weak vorticity formulation of the 2D Euler equation (see \cite{Schochet}). Here, for any $\phi\in C^\infty(\T^2)$,
  $$H_\phi(x,y)= \frac12 K(x-y)\cdot (\nabla\phi(x)- \nabla\phi(y)),\quad x,y\in \T^2, x\neq y,$$
in which $K$ is the Biot--Savart kernel on $\T^2$. We shall set $H_\phi(x,x)=0$ for all $x \in \T^2$. Note that $\mu^{a,b}$ is absolutely continuous with respect to the enstrophy measure $\mu$, with a density function bounded by $1/Z^{a,b}$, where $Z^{a,b}= \mu(\{\H \in [a,b]\}) >0$. Thus by \cite[Theorem 10 and Definition 11]{F1}, the nonlinear term $\<\omega_s\otimes \omega_s, H_\phi\>$ is well defined.

\begin{remark}
We recall that Cipriano showed in \cite[Theorem 4.1]{Cipriano} the existence of solutions to 2D Euler equations with given energy $a\in \R$, as long as the density function of $\H$ is positive at $a$. It is interesting to prove the same result by letting $b\to a$ in the above definition. The key ingredient is to show uniform estimates (independent of $a$ and $b$) of the type in Lemma \ref{lem-energy-solu} (without the parameter $N$). For the moment we do not know how to do this.
\end{remark}

Now we state our main result in this part.

\begin{theorem}\label{thm-energy-solu}
There exists a probability space $(\Theta,\mathcal F, \P)$ with the following properties.
\begin{itemize}
\item[\rm (i)] There exists a stochastic process $\omega:[0,T]\times \Theta\to H^{-1-}(\T^2)$ such that it is an energy conditional solution of the 2D Euler equations in the sense of Definition \ref{def}.
\item[\rm (ii)] On $(\Theta,\mathcal F, \P)$ one can define a subsequence of random point vortex system which converges $\P$-a.s. in $C\big( [0,T], H^{-1-}(\T^2)\big)$ to the solution of point {\rm (i)}.
\item[\rm (iii)] On $(\Theta,\mathcal F, \P)$ one can define a subsequence of functions $\omega^{(n)}(\theta,t,x),\, (\theta,t,x)\in \Theta \times [0,T]\times \T^2$, such that for $\P$-a.s. $\theta\in \Theta$, the functions $(t,x)\mapsto \omega^{(n)}(\theta,t,x)$ are $L^\infty$-solutions of 2D Euler equations, and converge to $\omega_\cdot(\theta)$ in $C\big( [0,T], H^{-1-}(\T^2)\big)$.
\end{itemize}
\end{theorem}

Recall that $\mu^{a,b}= (Z^{a,b})^{-1} {\bf 1}_{\{:\mathcal H: \in [a,b]\}} \mu$, it may seem that the above result follows from \cite[Theorem 25]{F1}. However, the initial density function in the present case is not continuous, thus our result is not a direct consequence of \cite[Theorem 25]{F1}. A careful investigation of the proof in \cite[Section 5]{F1} reveals that the continuity of density function was only used there to show that the normalizing constants $C_N$ tend to 1 as $N\to \infty$ (see the arguments below \cite[Lemma 29]{F1}). Since we have already shown in \eqref{eq-2} the convergence $Z_N^{a,b}\to Z^{a,b}$, we can follow the ideas in \cite{F1} to prove Theorem \ref{thm-energy-solu}.  In the sequel we present the main preparations which are needed in the proof.

Let $N\geq 2$ be fixed, we consider the point vortex dynamics on $\T^2$:
  \begin{equation}\label{vortex-dynamics}
  \frac{\d X^{i,N}_t}{\d t}= \frac1{\sqrt N} \sum_{j=1}^N a_j K\big(X^{i,N}_t -X^{j,N}_t \big), \quad i=1,\ldots, N
  \end{equation}
with the vortex intensities $(a_1,\ldots, a_N)\in (\R\setminus \{0\})^N$ and initial positions $\big(X^{1,N}_0, \ldots, X^{N,N}_0 \big)\in (\T^2)^N \setminus \Delta_N$, where $\Delta_N= \big\{ (x_1, \ldots, x_N)\in (\T^2)^N: \exists\, i\neq j \mbox{ such that } x_i= x_j\big\}$ is the generalized diagonal. It is well known that, for ${\rm Leb}_{\T^2}^{\otimes N}$-a.e. initial condition $\big(X^{1,N}_0, \ldots, X^{N,N}_0 \big)\in (\T^2)^N \setminus \Delta_N$, the above system of equations has a global solution, that is, the vortex points do not collapse, cf. \cite[Section 4.4]{MP}. Therefore, we can define the vorticity field
  $$\omega^N_t = \frac1{\sqrt N} \sum_{i=1}^N a_i \delta_{X^{i,N}_t},\quad t\geq 0,$$
which satisfies, for any $\phi\in C^\infty(\T^2)$,
  \begin{equation}\label{weak-vorticity-form}
  \big\<\omega^N_t, \phi\big\>= \big\<\omega^N_0, \phi\big\> + \int_0^t \big\<\omega^N_s\otimes \omega^N_s, H_\phi\big\>\,\d s, \quad \mbox{for all } t\geq 0.
  \end{equation}
We mention that an interesting model involving the creation and damping of point vortices is studied in the recent work \cite{Grotto}.

We shall consider the point vortex dynamics with random intensities and random initial conditions. Thus, on a probability space $(\Theta, \mathcal F, \P)$, let $\{\xi_i\}_{i\geq 1}$ be an i.i.d. sequence of random variables with the standard Gaussian distribution $N(0,1)$, and $\{X^i_0\}_{i\geq 1}$ an i.i.d. sequence of random variables with uniform distribution on $\T^2$; assume the two families are independent. Note that the measures $\lambda_N$ and $\mu_N$ defined in Section 1 are the laws of $\big((\xi_1, X^1_0), \ldots, (\xi_N, X^N_0)\big)$ and of $\omega^N_0= \frac1{\sqrt N} \sum_{i=1}^N \xi_i \delta_{X^{i}_0}$, respectively.

For the reader's convenience, we recall the following result which is taken from \cite[Proposition 22]{F1} (the last property is not mentioned there, but it is obvious since the Hamiltonian is invariant for point vortices).

\begin{proposition}\label{prop-energy}
Consider the point vortex dynamics \eqref{vortex-dynamics} with random intensities $(\xi_1, \ldots, \xi_N)$ and random initial positions $( X^1_0, \ldots, X^N_0)$ distributed as $\lambda_N$. $\P$-almost surely, the dynamics $\big( X^{1,N}_t, \ldots, X^{N,N}_t\big)$ is well defined in $(\T^2)^N \setminus \Delta_N$ for all $t\geq 0$, and $\big( \big(\xi_1, X^{1,N}_t\big), \ldots, \big(\xi_N, X^{N,N}_t\big) \big)$ has the invariant law $\lambda_N$. The associated measure-valued vorticity $\omega^N_t$ is a stochastic process with the invariant marginal law $\mu_N$; moreover, $\P$-a.s., $\omega^N_t$ satisfies \eqref{weak-vorticity-form} and
  \begin{equation}\label{prop-energy.1}
  \mathcal H_N \big(\omega^N_t \big)= \mathcal H_N\big( \big(\xi_1, X^{1,N}_t\big), \ldots, \big(\xi_N, X^{N,N}_t\big) \big) = \mathcal H_N \big(\omega^N_0 \big), \quad \mbox{for all } t\geq 0.
  \end{equation}
\end{proposition}

Next we confine the point vortex dynamics to those initial configurations with energy belonging to the interval $[a,b]$. To this end, we introduce the conditional probability measures on $(\Theta, \mathcal F)$:
  $$\P_N^{a,b}= \frac{{\bf 1}_{\{\mathcal H_N(\omega^N_0)\in [a,b]\}}}{\P(\{ \mathcal H_N(\omega^N_0)\in [a,b] \})} \P = \frac{{\bf 1}_{ \{\mathcal H_N(\omega^N_0) \in [a,b]\}}}{Z_N^{a,b}} \P.$$
The measure $\lambda_N^{a,b}$ defined in \eqref{flat-cond-measure} is the distribution on $(\R\times \T^2)^N$ of $\big( \big(\xi_1, X^{1}_0\big), \ldots, \big(\xi_N, X^{N}_0\big) \big)$ under $\P_N^{a,b}$, and we have $\mu_N^{a,b} = (\mathcal T_N)_\# \lambda_N^{a,b}$, where $\mathcal T_N$ is defined in \eqref{map}. From Proposition \ref{prop-energy} we deduce the following result.

\begin{proposition}\label{prop-conditional-energy}
Consider the point vortex dynamics \eqref{vortex-dynamics} with random intensities $(\xi_1, \ldots, \xi_N)$ and random initial positions $( X^1_0, \ldots, X^N_0)$ distributed as $\lambda_N^{a,b}$. Then, $\P_N^{a,b}$-a.s., for all $t\geq 0$, the dynamics $\big( X^{1,N}_t, \ldots, X^{N,N}_t\big)$ is well defined in $(\T^2)^N \setminus \Delta_N$, and $\big( \big(\xi_1, X^{1,N}_t\big), \ldots, \big(\xi_N, X^{N,N}_t\big) \big)$ has the invariant distribution $\lambda_N^{a,b}$. The associated measure-valued vorticity $\omega^N_t$ is a stochastic process with the invariant marginal law $\mu_N^{a,b}$; moreover, $\P_N^{a,b}$-a.s., $\omega^N_t$ satisfies \eqref{weak-vorticity-form} and
  $$\mathcal H_N \big(\omega^N_t \big) = \mathcal H_N \big(\omega^N_0 \big)\in [a,b], \quad \mbox{for all } t\geq 0.$$
\end{proposition}

\begin{proof}
Since the conditional probability measure $\P_N^{a,b}$ is absolutely continuous with respect to $\P$, the properties that hold $\P$-a.s. also hold almost surely with respect to $\P_N^{a,b}$. It remains to show that $\lambda_N^{a,b}$ is the invariant distribution of $\big( \big(\xi_1, X^{1,N}_t\big), \ldots, \big(\xi_N, X^{N,N}_t\big) \big)$. Once this is proved, we deduce that $\omega^N_t$ has the invariant marginal law $\mu_N^{a,b}$ since
  $${\rm law}\big(\omega^N_t \big) = (\mathcal T_N)_\# {\rm law}\big(\big(\xi_1, X^{1,N}_t\big), \ldots, \big(\xi_N, X^{N,N}_t\big) \big) = (\mathcal T_N)_\# \lambda_N^{a,b} =\mu_N^{a,b}.$$

To simplify the notations, we write $\xi= (\xi_1, \ldots, \xi_N)$ and $X^N_t= \big(X^{1,N}_t, \ldots, X^{N,N}_t \big)$. For any bounded measurable function $F: (\R\times \T^2)^N\to \R$, by the definition of $\P_N^{a,b}$ and \eqref{prop-energy.1},
  $$\aligned \int_\Theta F\big(\xi, X^N_t\big)\,\d \P_N^{a,b} &= \frac1{Z_N^{a,b}} \int_\Theta F\big(\xi, X^N_t\big) {\bf 1}_{ \{\mathcal H_N(\omega^N_0) \in [a,b]\}} \,\d\P \\
  &= \frac1{Z_N^{a,b}} \int_\Theta F\big(\xi, X^N_t\big) {\bf 1}_{ \{\mathcal H_N (\xi, X^N_t) \in [a,b]\}} \,\d\P \\
  &= \frac1{Z_N^{a,b}} \int_\Theta F \big(\xi, X^N_0\big) {\bf 1}_{ \{\mathcal H_N (\xi, X^N_0) \in [a,b]\}} \,\d\P,
  \endaligned$$
where the last step is due to the fact that, under $\P$, $\big(\xi, X^N_t\big)$ has the invariant distribution $\lambda_N$. Therefore,
  \[\int_\Theta F\big(\xi, X^N_t\big)\,\d \P_N^{a,b} = \int_\Theta F \big(\xi, X^N_0\big) \,\d \P_N^{a,b},\]
which implies that, under the conditional probability measure $\P_N^{a,b}$, $\big(\xi, X^N_t\big)$ has the same law as $\big(\xi, X^N_0\big)$, i.e. $\lambda_N^{a,b}$.
\end{proof}

To emphasize the dependence on the parameters $a,b$, we denote by $\omega^N_{a,b}(\cdot)$ the associated measure-valued vorticity field obtained in Proposition \ref{prop-conditional-energy}. The next lemma gives useful estimates on $\omega^N_{a,b}(\cdot)$.

\begin{lemma}\label{lem-energy-solu}
Let $N_0$ be large enough such that $Z_0:= \inf_{N\geq N_0} Z_N^{a,b}>0$ and $f:\T^2\times \T^2\to \R$ be symmetric, bounded and measurable. Then for all $p\geq 1$ and $\delta>0$, there are constants $C_p, C_{p,\delta}>0$ such that for all $N\geq N_0$,
  $$\E_{\P_N^{a,b}} \big[\big\<\omega^N_{a,b}(t)\otimes \omega^N_{a,b}(t), f \big\>^p\big] \leq C_p \|f\|_\infty^p / Z_0, \quad \E_{\P_N^{a,b}} \big[\big\|\omega^N_{a,b}(t)\big\|_{H^{-1-\delta}}^p \big] \leq C_{p,\delta}/ Z_0.$$
Moreover, if $f(x,x)=0$ for all $x\in\T^2$, then
  $$\E_{\P_N^{a,b}} \big[\big\<\omega^N_{a,b}(t)\otimes \omega^N_{a,b}(t), f \big\>^2\big] \leq \frac2{Z_0} \int_{\T^2\times \T^2} f(x,y)^2\,\d x\d y.$$
\end{lemma}

\begin{proof}
Note that
  $$\E_{\P_N^{a,b}} \big[\big\<\omega^N_{a,b}(t)\otimes \omega^N_{a,b}(t), f \big\>^p\big] = \frac1{Z_N^{a,b}} \E_{\P} \Big[\big\<\omega^N_t\otimes \omega^N_t, f \big\>^p {\bf 1}_{\{\mathcal H_N(\omega^N_0) \in [a,b]\}}\Big], $$
the first estimate follows from \cite[Lemma 23]{F1}. The proofs of the others are similar.
\end{proof}

With these preparations in hand, we can complete the proof of Theorem \ref{thm-energy-solu}. More precisely, let $Q^N$ be the law of the process $\big\{ \omega^N_{a,b}(t) \big\}_{t\in [0,T]}$ on $\mathcal X= C\big([0,T], H^{-1-}(\T^2) \big)$. Using the equation \eqref{weak-vorticity-form} and the estimates in Lemma \ref{lem-energy-solu}, we can show that the family $\{Q^N\}_{N\geq N_0}$ is tight on $\mathcal X$, see the beginning part of \cite[Section 4.2]{F1} for details. By Prohorov's theorem, there is a subsequence $\{N_k\}_{k\geq 1}$ such that $Q^{N_k}$ converges weakly to some probability measure $Q$ on $\mathcal X$. Skorokhod's theorem implies that there exist a probability space $\big( \tilde\Theta, \tilde{\mathcal F},\tilde\P\big)$ and processes $\tilde\omega^{N_k}$ and $\tilde\omega$, with trajectories in $\mathcal X$, such that their laws are $Q^{N_k}$ and $Q$, respectively; and $\tilde\P$-a.s., $\tilde\omega^{N_k}$ converges to $\tilde\omega$ in the topology of $\mathcal X$. Moreover, the processes $\tilde\omega^{N_k}$ can be represented as a sum of Dirac deltas:
  $$\tilde\omega^{N_k}_t= \frac1{\sqrt{N_k}} \sum_{i=1}^{N_k} \tilde\xi_i \delta_{\tilde X^{i,N_k}_t}, \quad t\in [0,T],$$
where $\big(\big(\tilde\xi_1, \tilde X^{i,N_k}_t\big),\ldots, \big(\tilde\xi_{N_k}, \tilde X^{N_k,N_k}_t\big)\big)$ is a random vector with the invariant law $\lambda_{N_k}^{a,b}$ for all $t\in[0, T]$, and it solves the point vortex dynamics \eqref{vortex-dynamics}, see \cite[Lemma 28]{F1} for the detailed proof.

Next we prove the law of $\tilde\omega_t$ is the energy conditional measure $\mu^{a,b}$ for all $t\in [0,T]$. For any $F\in C_b\big( H^{-1-}(\T^2) \big)$, since $\tilde\omega^{N_k}$ converges $\tilde\P$-a.s. to $\tilde\omega$ in the topology of $\mathcal X$,
  $$\aligned
  \int_{\tilde \Theta} F(\tilde \omega_t)\,\d\tilde\P &= \lim_{k\to \infty} \int_{\tilde \Theta} F\big(\tilde \omega^{N_k}_t \big) \,\d\tilde\P = \lim_{k\to \infty} \int_{H^{-1-}(\T^2)} \! F(\omega)\,\d \mu_{N_k}^{a,b}(\omega) = \int_{H^{-1-}(\T^2)} \! F(\omega)\,\d \mu^{a,b}(\omega),
  \endaligned$$
where in the last step we have used the weak convergence of $\mu_N^{a,b}$ to $\mu^{a,b}$ proved in Section 3.

Finally, using again the estimates in Lemma \ref{lem-energy-solu} and repeating the arguments below \cite[Lemma 28]{F1}, we can show that $\tilde\omega_\cdot$ satisfies the weak vorticity formulation \eqref{weak-vorticity} of the 2D Euler equation. Summarizing the above discussions, we have proved the first two assertions of Theorem \ref{thm-energy-solu}. The last assertion is proved in the same way as the end of \cite[Section 4.2]{F1}.

\section{Structures and intermediate regimes}

In classical Onsager theory the microcanonical measure is defined as the
uniform measure on configurations $\left(  x_{1},\ldots ,x_{N}\right)  $ such
that
  \begin{equation}\label{microcan Onsager}
  \sum_{i\neq j}\xi_{i} \xi_{j}\log\frac{1}{\left\vert x_{i}-x_{j} \right\vert }\sim N^{2} a
  \end{equation}
for some value of $a>0$ (in this section we heuristically write $a$ instead of $[a,b]$ since, for $a>0$, it is the value of $a$ which plays a practical role, independently of $b$). For typical configurations $\left(  x_{1}%
,\ldots ,x_{N}\right)  $, when $N$ is large, the empirical measure%
\[
\frac{1}{N}\sum_{i=1}^{N} \xi_{i} \delta_{x_{i}}%
\]
is close to the solutions of a certain mean field equation (Onsager theory).
There is a natural explanation, for $a\gg 0$:\ in order to have
(\ref{microcan Onsager}) we need roughly $N^{2}$ terms in the sum $\sum_{i\neq
j}$ with value $\xi_{i} \xi_{j}\log\frac{1}{\left\vert x_{i}%
-x_{j}\right\vert }$ close to $a$ (this argument is very rough). The ``only''
way to reach such result is to group positive vortices together, all very
close to each other, and similarly for the negative ones, with the two clusters
not so close to each other:\ roughly $\left(  N/2\right)  ^{2}$ terms will be
positive and close to $a$ (those corresponding to positive pairs), other
$\left(  N/2\right)  ^{2}$ terms will be positive as well and close to $a$
(those corresponding to negative pairs), and the remaining pairs, composed of
vortices of opposite signs, have small value of $\xi_{i} \xi_{j}\log
\frac{1}{\left\vert x_{i}-x_{j}\right\vert }$ because the two points belong to
clusters which are relatively far from each other.

In our ``white noise'' model, the microcanonical measure corresponds to the constraint
  \begin{equation}\label{WN-regime}
  \sum_{i\neq j} \xi_{i} \xi_{j}\log\frac{1}{\left\vert x_{i}-x_{j} \right\vert }\sim N a.
  \end{equation}
The typical configurations $\left(  x_{1},\ldots ,x_{N}\right)  $, for large $N$, have the renormalized empirical measure%
  \[ \frac{1}{\sqrt{N}}\sum_{i=1}^{N} \xi_{i}\delta_{x_{i}} \]
close to white noise conditioned to renormalized energy equal to $a$. In Figure 1 below, we show the histogram of the interaction energy of 200 point vortices (its features do not change by increasing the number of vortices). It shows the typical values of ``$a$'' in formula \eqref{WN-regime}. They are very small and the corresponding configurations are quite disordered, opposite to the structures of Onsager theory and coherently with the white noise limit. The theoretical energy spectrum of \textit{free} white noise ensemble (not constrained by the energy) decays as $k^{-1}$, opposite to the predicted decay $k^{-5/3}$ of inverse stationary 2D turbulence. The question then is the decay of the spectrum for the microcanonical ensemble, especially for large values of $a$, when we expect some degree of clustering of the vortices and then, potentially, the emergence of a more interesting spectrum.

\begin{figure}[htbp]
\centering
\includegraphics[width=15cm, height=7.5cm]{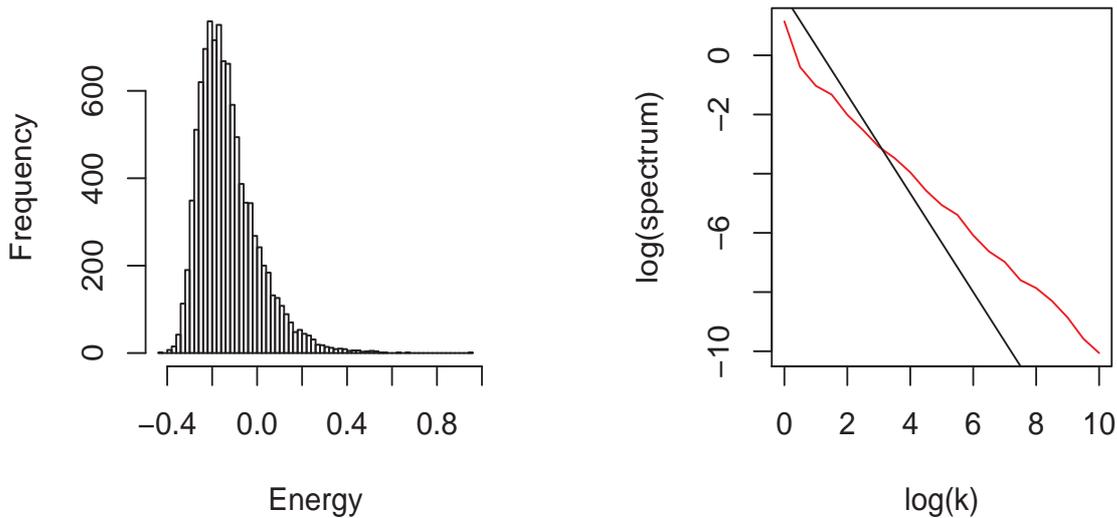}
\caption{Left: the histogram of 10,000 samples of the interaction energies of 200 uniformly distributed point vortices. Right: the curve is the spectrum function computed from the 11 samples with largest interaction energies, while the straight line has the slope $-5/3$.}
\label{fig-1}
\end{figure}

It is very difficult to compute theoretically the spectrum function of the microcanonical measure \eqref{microcanonical-meas}, thus we do some numerical simulations. We generate 10,000 samples of uniformly distributed point vortices; each sample consists of 200 vortices, in which half of them have intensity $1/\sqrt{200}$ and the rest $-1/\sqrt{200}$. We single out the 11 samples which have the largest interaction energies ($a= 0.51$ in this case), and compute their average spectrum function. The results are shown in Figure \ref{fig-1}. It shows that the slope of the spectrum is still very close to $-1$ like in the case of the free ensemble, far from $-5/3$.

Deviations of the spectrum slope from the flat value $-1$ are due to clusterization of point vortex configuration. To prove numerically this claim we proceed as follows: we produce artificially an initial condition with small clusters and then let it evolve by point vortex dynamics. We do not have a theorem of convergence to equilibrium but hope that after some time the configuration is more typical for the microcanonical ensemble. Precisely, we generate a point vortex configuration which, apart from some uniformly distributed point vortices, contains small clusters with 2, 4 and 8 vortices (these numbers are chosen for convenience). The clusters have uniformly distributed centers and their diameters are of the order 0.01. To get a smoother spectrum function, we produce 10 such samples (with average energy 1.364966) and compute the averaged spectrum function, which is shown by the thin line on the right of Figure \ref{fig-2}. We see that it is close to the line with slope $-5/3$ in a certain range of $\log(k)$. We take these special configurations as initial conditions and run the dynamics (use the Heun algorithm), with a small time step $h=0.0001$. In Figure \ref{fig-2}, we show the vortex distribution of one of the samples after 120,000 steps of evolution: $+$ and $\circ$ represent vortices of positive and negative intensity, respectively. The graph of the final spectrum function is shown by the dashed line on the right of Figure \ref{fig-2}, which, on the range $\log(k)\in [1,3]$, has the approximative slope $-1.775$. Opposite to the cases considered in Figure \ref{fig-1}, here we find a slope considerably different from $-1$ and in the direction of $-5/3$.

\begin{figure}[htbp]
\centering
\includegraphics[width=15cm, height=7.5cm]{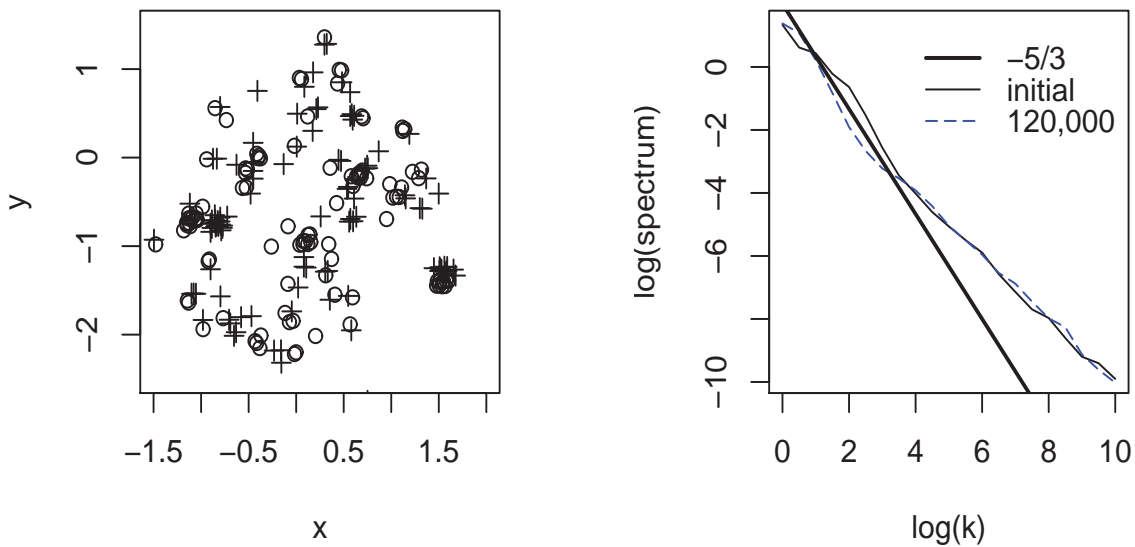}
\caption{Left: the configuration of (major part of) vortices after 120,000 steps of evolution. Right: spectrum functions before and after running dynamics.}
\label{fig-2}
\end{figure}

The question then is how to obtain spontaneously some degree of local clusterization, from an invariant measure and in particular from a microcanonical ensemble. Compared to turbulence, it seems that the two regimes of Onsager and conditional white noise are two ``extremes''. Turbulence is in the middle: typical configurations are not so uniformly distributed as in the white noise case, they have locally a great degree of clustering. But only locally, at small scale, not globally as the two big clusters of the Onsager case. Thus in the turbulence regime we expect that each vortex interacts neither with all those of the same sign (as in (\ref{microcan Onsager}))\ nor only with very few of the same sign (as in \eqref{WN-regime}), but with an intermediate amount.

A natural microcanonical condition is therefore
  \[\sum_{i\neq j} \xi_{i} \xi_{j}\log\frac{1}{\left\vert x_{i}-x_{j} \right\vert }\sim c(N) e \]
for some
  \[ N\ll c(N)\ll N^2. \]
The mathematical question then is whether it is possible to study the limit as $N\rightarrow \infty$ of this intermediate regime. For finite $N$ the microcanonical measure with normalizing constant $c\left(  N\right)  $  is invariant for Euler dynamics, but we do not know a corresponding invariant measure obtained as $N\rightarrow\infty$.  We leave this question open but hope the clarifications of this work help to address the question.

\bigskip

\noindent \textbf{Acknowledgements.} The second author is grateful to the financial supports of the grant ``Stochastic models with spatial structure'' at the Scuola Normale Superiore di Pisa, the National Natural Science Foundation of China (Nos. 11571347, 11688101), and the Youth Innovation Promotion Association, CAS (2017003).


\begin{thebibliography}{99}

%\bibitem{AC} S. Albeverio, A. B. Cruzeiro, Global flows with invariant (Gibbs) measures for Euler and Navier--Stokes two-dimensional fluids. \textit{Comm. Math. Phys.} \textbf{129} (1990), 431--444.

\bibitem{ARH} S. Albeverio, M. Ribeiro de Faria, R. H{\o}egh-Krohn, Stationary measures for the periodic Euler flow in two dimensions. \emph{J. Statist. Phys.} \textbf{20} (1979), 585--595.

\bibitem{AF2} S. Albeverio and B. Ferrario, Some Methods of Infinite Dimensional Analysis in Hydrodynamics: An Introduction, In SPDE in Hydrodynamic: Recent Progress and Prospects, G. Da Prato and M. R\"{o}ckner Eds., CIME Lectures, Springer--Verlag, Berlin 2008.

\bibitem{BPP} G. Benfatto, P. Picco, M. Pulvirenti, On the invariant measures for the two-dimensional Euler flow. \emph{J. Statist. Phys.} \textbf{46} (1987), no. 3--4, 729--742.

\bibitem{BV} F. Bouchet, A. Venaille, Statistical mechanics of two-dimensional and geophysical flows. \emph{Phys. Rep.} \textbf{515} (2012), no. 5, 227--295.

\bibitem{CLMP-1} E. Caglioti, P. L. Lions, C. Marchioro, M. Pulvirenti, A special class of stationary flows for two-dimensional Euler equations: a statistical mechanics description. \emph{Comm. Math. Phys.} \textbf{143} (1992), no. 3, 501--525.

\bibitem{CLMP-2} E. Caglioti, P. L. Lions, C. Marchioro, M. Pulvirenti, A special class of stationary flows for two-dimensional Euler equations: a statistical mechanics description. II. \emph{Comm. Math. Phys.} \textbf{174} (1995), no. 2, 229--260.

\bibitem{Cipriano} F. Cipriano. The two-dimensional Euler equation: a statistical study. \emph{Comm. Math. Phys.} \textbf{201} (1999),  no. 1, 139--154.

\bibitem{Eyink} G. L. Eyink, H. Spohn, Negative-temperature states and large-scale, long-lived vortices in two-dimensional turbulence. \emph{J. Stat. Phys.} \textbf{70} (1993), no. 3-4, 833-886.

\bibitem{F1} F. Flandoli, Weak vorticity formulation of 2D Euler equations with white noise initial condition. \emph{Comm. Partial Differential Equations} \textbf{43} (2018), 1102--1149.

% \bibitem{F2} F. Flandoli, Random Initial Conditions and Noise in Ordinary and Partial Differential Equations. Lecture notes at the Scuola Normale Superiore di Pisa, 2018.

\bibitem{FL-3} F. Flandoli, D. Luo, Convergence of transport noise to Ornstein--Uhlenbeck for 2D Euler equations under the enstrophy measure. \emph{Ann. Probab.} (2019), accepted.

\bibitem{FS} F. Flandoli, M. Saal, mSQG equations in distributional spaces and point vortex approximation. arXiv:1812.05361v1.

\bibitem{GR} C. Geldhauser, M. Romito, Limit theorems and fluctuations for point vortices of generalized Euler equations, arXiv:1810.12706.

\bibitem{Grotto} F. Grotto, Stationary Solutions of Damped Stochastic 2-dimensional Euler's Equation, arXiv:1901.06744.

\bibitem{GrottoRomito} F. Grotto, M. Romito, A Central Limit Theorem for Gibbsian Invariant Measures of 2D Euler Equation, arXiv:1904.01871v1.

\bibitem{Lions} P. L. Lions, On Euler equations and statistical physics. Cattedra Galileiana. [Galileo Chair] \emph{Scuola Normale Superiore, Classe di Scienze, Pisa,} 1998.

\bibitem{Malliavin} P. Malliavin. Universal Wiener space. \emph{Barcelona Seminar on Stochastic Analysis (St. Feliu de Gu¨ªxols, 1991),} 77--102, Progr. Probab., 32, \emph{Birkh\"auser, Basel,} 1993.

\bibitem{MP} C. Marchioro, M. Pulvirenti, Mathematical theory of incompressible nonviscous fluids, volume 96 of Applied Mathematical Sciences, Springer--Verlag, New York, 1994.

\bibitem{Neri} C. Neri, Statistical mechanics of the $N$-point vortex system with random intensities on a bounded domain. \emph{Ann. Inst. H. Poincar\'e Anal. Non Lin\'eaire} \textbf{21} (2004), no. 3, 381--399.

\bibitem{Schochet} S. Schochet, The weak vorticity formulation of the 2-D Euler equations and concentration-cancellation. \emph{Comm. Partial Differential Equations} \textbf{20} (1995), no. 5--6, 1077--1104.

\bibitem{TER} H. Touchette, R. S. Ellis, B. Turkington, An introduction to the thermodynamic and macrostate levels of nonequivalent ensembles. \emph{Phys. A} \textbf{340} (2004), no. 1--3, 138--146.

\end{thebibliography}
\end{document}